\documentclass[11pt]{article}
\usepackage[margin=1in, headheight=22.54448pt]{geometry}
\usepackage[bottom, flushmargin]{footmisc}
\RequirePackage{fix-cm}
\usepackage{fancyhdr}
\usepackage{hyperref} 
\usepackage{comment} 
\usepackage{enumitem} 
\usepackage{graphicx} 
\usepackage{pdfpages} 
\usepackage{amsmath} 
\usepackage{amssymb} 
\usepackage{amsthm}
\usepackage{eurosym}
\usepackage{color} 
\usepackage{xcolor} 
\usepackage[edges]{forest}
\usepackage[outline]{contour}
\usepackage{booktabs}

\hypersetup{
	colorlinks,
	linkcolor= black, 
	urlcolor=black,  
	citecolor= black 
}
\usepackage[round]{natbib}
\usepackage{har2nat}

\usepackage{tikz}
\usepackage[normalem]{ulem}
\usepackage{pgfplots}
\usepackage{pgfplotstable}
\usepackage{setspace}

\usepackage{multirow}
\usepackage{multicol}
\usepackage{array}

\newcommand{\indep}{\!\perp\!\!\!\perp}

\makeatletter
\newcommand*{\centernot}{%
  \mathpalette\@centernot
}
\def\@centernot#1#2{%
  \mathrel{%
    \rlap{%
      \settowidth\dimen@{$\m@th#1{#2}$}%
      \kern.5\dimen@
      \settowidth\dimen@{$\m@th#1=$}%
      \kern-.5\dimen@
      $\m@th#1\not$%
    }%
    {#2}%
  }%
}
\makeatother

\renewcommand{\cite}[1]{\citeyear{#1}}
\theoremstyle{plain}

\newtheorem{lemma}{Lemma}
\newtheorem{theorem}{Theorem}
\newtheorem{proposition}{Proposition}

\newtheorem{assumption}{Assumption}

\theoremstyle{definition}
\newtheorem{example}{Example}
\newtheorem{remark}{Remark}

\long\def\symbolfootnote[#1]#2{\begingroup%
\def\thefootnote{\fnsymbol{footnote}}\footnote[#1]{#2}\endgroup}

\usepackage{footnote}
\makesavenoteenv{tabular}
\usepackage{fancyhdr}
\newcommand{\documenttitle}{Thesis}

\renewcommand{\max}{\operatornamewithlimits{max}}
\renewcommand{\min}{\operatornamewithlimits{min}}

\pgfplotsset{compat=1.15}

\title{Asymptotic Theory for Two-Way Clustering}
\author{Luther Yap \thanks{Department of Economics, Princeton University. Email: \texttt{lyap@princeton.edu}. I thank Michal Koles{\'a}r, David Lee, and Ulrich M\"{u}ller for helpful comments and suggestions. This paper supersedes an earlier paper circulated as ``General Conditions for Valid Inference in Multi-Way Clustering".}}
\date{\today}

\begin{document}
\pagenumbering{arabic}
\maketitle


\begin{abstract}
This paper proves a new central limit theorem for a sample that exhibits two-way dependence and heterogeneity across clusters. Statistical inference for situations with both two-way dependence and cluster heterogeneity has thus far been an open issue. The existing theory for two-way clustering inference requires identical distributions across clusters (implied by the so-called separate exchangeability assumption). Yet no such homogeneity requirement is needed in the existing theory for one-way clustering. The new result therefore theoretically justifies the view that two-way clustering is a more robust version of one-way clustering, consistent with applied practice. In an application to linear regression, I show that a standard plug-in variance estimator is valid for inference. 
\end{abstract}

\clearpage
\section{Introduction}
Clustering standard errors on multiple dimensions is common and attractive in applied econometrics because it allows observations to be dependent whenever they share a cluster on any dimension. 
Though more broadly applicable, a common instance of two-way clustering is in linear regressions, where a researcher wants to do inference on the coefficient of interest when the residual is two-way clustered. 
The variance estimator proposed by \citet{cameron2011robust} (henceforth CGM) has thus been widely applied to contexts with such two-way dependence.\footnote{CGM has 3886 citations on Google Scholar at the time of writing.}
For instance, \citet{nunn2011slave} clustered on ethnic group and district when studying the effect of slave trade on trust; \citet{michalopoulos2013pre} clustered on country and ethnolinguistic family when studying the effect of pre-colonial institutions on development; \citet{jackson2018test} clustered on teacher and student when studying the effect of the teacher on students' skill; \citet{neumark2019harder} clustered on resume and job ad when studying the effect of age on getting a call-back. 
The existing justification for the asymptotic validity of the CGM estimator and other inference procedures in two-way clustering (e.g., \citet{mackinnon2021wild}; \citet{davezies2021empirical};  \citet{menzel2021bootstrap}) relies on separate exchangeability, which implies homogeneity of clusters, a restriction that is not required in one-way clustering. This paper provides sufficient general conditions for valid inference in two-way clustering by proving that, even with cluster heterogeneity, a central limit theorem holds, and the CGM variance estimator is consistent. 

An environment with two-way clustering permits dependence whenever observations share at least one cluster. To fix ideas, consider \citet{jackson2018test}: observations of the same student or of the same teacher are plausibly correlated, but two observations of different students and different teachers are assumed to be independent.\footnote{This setting permits more general dependence structures than one-way clustering. If there is one-way clustering by student, then two observations from different students are automatically independent. In two-way clustering, two observations from different students are not necessarily independent because they may share the same teacher.} 
The CGM variance estimator accommodates such dependence, and a subsequent literature provided a theoretical basis for its validity: \citet{mackinnon2021wild} obtained sufficient conditions for validity of the CGM estimator in regression models; \citet{davezies2021empirical} obtained analogous results for empirical processes. \citet{menzel2021bootstrap} also showed the validity of a bootstrap procedure for two-way clustering that is robust to asymptotic non-normalities.\footnote{\citet{menzel2021bootstrap} pointed out that a purely interactive data-generating processes unique to two-way dependence has an asymptotic distribution that is not normal. Section \ref{sec:CLT} will consider this process and show how the assumptions of this paper rule it out.} The theoretical basis for inference thus far relies on separate exchangeability, the assumption that random variables are exchangeable on either clustering dimension, though not necessarily both.

As noted by \citet{mackinnon2021wild}, however, separate exchangeability implies identical marginal distributions. Separate exchangeability in the student-teacher example thus implies the random variables for all students must be drawn from the same distribution, including students of different cohorts over time. 
As \citet[p.~146]{wooldridge2010econometric} notes in the discussion of pooled data in his graduate textbook, distributions of variables tend to change over time, so the identical distribution assumption is not usually valid.
In other examples, separate exchangeability implies that countries \citep{michalopoulos2013pre} and jobs \citep{neumark2019harder} are identically distributed.
Applied researchers surely would want size to be controlled in such heterogeneous environments, but the existing theories that rely on separate exchangeability do not imply this result. 
Further, in linear regressions with regressor $X_i$ and residual $u_i$, asymptotic theory is applied to $X_i u_i$. Separate exchangeability of the product implies that the regressors must also be separately exchangeable, which is not plausible when the regressors include a time trend, say.


In contrast, existing asymptotic theory on one-way clustering (e.g., \citet{hansen2019asymptotic}; \citet{djogbenou2019asymptotic}) allows the distribution of the random variable to be heterogeneous over clusters. Since the only available conditions for the validity of two-way clustering require separate exchangeability, the literature lacks conditions for two-way clustering that generalize one-way clustering and permit heterogeneity over clusters. This paper fills the gap, and thus justifies two-way clustering as a more robust version of one-way clustering. In particular, when the conditions of the central limit theorem hold, the variance estimator for two-way clustering asymptotically converges to the true variance for one-way clustering when the researcher clusters on more dimensions than required. 

\begin{example}
To illustrate separate exchangeability, consider an additive random effects model. Individual $i$ who belongs to cluster $g(i)$ on the $G$ dimension and cluster $h(i)$ on the $H$ dimension is characterized by a random variable $W_i$ generated from $W_i = \alpha_{g(i)} + \gamma_{h(i)} + \varepsilon_i$, where cluster-specific $\alpha_1,\dots,\alpha_g,\dots,\alpha_G, \gamma_1,\dots, \gamma_h,\dots,\gamma_H$ and individual-specific $\varepsilon_1,\dots,\varepsilon_i,\dots,\varepsilon_n$ are mutually independent. If we assume separate exchangeability, then $\alpha_g$, $\gamma_h$, and $\varepsilon_i$ are iid.\footnote{To see this, for individuals $i$ and $j$ where $g(i) \ne g(j)$, $h(i) = h(j)=h$, separate exchangeability implies $\alpha_{g(i)} + \gamma_h + \varepsilon_i \stackrel{d}{=} \alpha_{g(j)} + \gamma_h + \varepsilon_j$. Since $\alpha_g, \gamma_h$ and $\varepsilon_i$ are independent, $\varepsilon_i \stackrel{d}{=} \varepsilon_j$ and $\alpha_g \stackrel{d}{=} \alpha_{g'}$.} In contrast, under one-way cluster asymptotics, the cluster-specific error $\alpha_g$ need not be identically distributed. The general conditions provided in this paper permit valid inference even when $\alpha_g, \gamma_h, \varepsilon_i$ are not identically distributed in this model.
\end{example}

The main result is a central limit theorem for two-way clustering with heterogeneous cluster sizes and distributions. This result is proven using Stein's method. It adapts the strategy from \citet{ross2011fundamentals} Theorem 3.6: I first derive an upper bound on the distance between the distribution of a pivotal statistic and the standard normal, then show that this distance converges to zero asymptotically. This proof strategy hence yields intermediate results on non-asymptotic Berry-Esseen type bounds that provide worst-case bounds on the quality of approximation between the pivotal statistic and the standard normal, which may be of independent interest. I apply the theorem to a simple setting of a linear regression, but it is more broadly applicable to many other econometric procedures that exhibit a similar clustering structure.

This paper contributes to the literature on multi-way clustering and Stein's method. 
This paper differs from the existing literature on multi-way clustering (e.g., \citet{mackinnon2021wild}; \citet{davezies2021empirical};  \citet{menzel2021bootstrap}; \citet{chiang2022standard}; \citet{chiang2023using}) in that it does not rely on separate exchangeability.
Stein's method has been applied to other contexts such as two-way fixed effects \citep{verdier2020estimation}, and spillover effects (e.g., \citet{chin2018central}, \citet{leung2022rate} and \citet{braun2023estimation}). 
Unlike the aforementioned papers, this paper speaks directly to multi-way clustering, and it makes a modification to the proof of \citet{ross2011fundamentals} Theorem 3.6 to obtain the result instead of applying the theorem directly. 

\section{Setting and Main Result}\label{sec:CLT}
\subsection{Setup}
Consider a setup with two-way clustering on dimensions $G$ and $H$ for random vectors $\{ W_i \}_{i=1}^n$, where $W_i := (W_{i1}, W_{i2}, \cdots, W_{iK})' \in \mathbb{R}^K$ and $i=1,\dots,n$ is the unit of observation. For example, $G$ could denote states and $H$ could denote industries. Clustering in more than two dimensions is possible, and derivations are entirely analogous. This section establishes a central limit theorem (CLT) for $\sum_i W_i$, as $n \rightarrow \infty$. Here and in the following, sums are over (subsets of) $\{ 1, 2, \dots, n \}$. For $C \in \{G ,H \}$, let $\mathcal{N}^C_c$ denote the set of observations in cluster $c$ on dimension $C$ --- this setup partitions the sample on the $C$ dimension.

Let $g(i)$ and $h(i)$ denote the cluster that observation $i$ belongs to on the $G$ and $H$ dimensions respectively. These cluster identities are nonstochastic and observed. Let $N^C_c := |\mathcal{N}^C_c|$ denote the cluster size for $C \in \{G, H \}$ and $N_{gh} := |\mathcal{N}^G_g \cap \mathcal{N}^H_h|$. These cluster sizes are allowed to be heterogeneous in a way that will be formalized in the assumptions below. $W_i$ is assumed to be independent of the joint distribution of $\{ W_j \}$ for $j \notin \mathcal{N}^G_{g(i)} \cup \mathcal{N}^H_{h(i)} =: \mathcal{N}_i$, i.e., when $i$ and $j$ do not share a cluster on either dimension. Hence, $\mathcal{N}_i$ is the set of observations that are arbitrarily dependent with $i$. This environment is stated as Assumption \ref{asmp:indep}. 

\begin{assumption} \label{asmp:indep} 
With $\mathcal{N}_i = \mathcal{N}^G_{g(i)} \cup \mathcal{N}^H_{h(i)}$,
\begin{enumerate} [topsep=0pt,label=(\alph*)]
    \item $W_i \indep \{ W_j \}_{j \notin \mathcal{N}_i}$ for all $i$.
    \item For observations $i,j $ and $k \in \mathcal{N}_i,l \in \mathcal{N}_j$ and all nonstochastic $\mu \in \mathbb{R}^K$, if $j,l \notin (\mathcal{N}_i \cup \mathcal{N}_k)$, then $Cov(\mu^\prime W_i W_k^\prime \mu, \mu^\prime W_j W_l^\prime \mu) =0$.
\end{enumerate}
\end{assumption}

While the dependence structure is implicitly described in the setup of many clustering papers (e.g., \citet{hansen2019asymptotic}; \citet{menzel2021bootstrap}), Assumption \ref{asmp:indep} makes the dependence structure explicit. Assumption \ref{asmp:indep}(a) is a dissociation assumption similar to definition 3.5 of \citet{ross2011fundamentals} required to apply Stein's method. Assumption \ref{asmp:indep}(b) is required because, for a scalar $W_i$, a crucial step of the proof requires $E[W_i W_j W_k W_l] = E[W_i W_k] E[W_j W_l]$ when $j,l$ do not share any cluster with $i,k$. Even when $W_i \indep (W_j, W_l)$ and  $W_k \indep (W_j, W_l)$, we cannot conclude that $E[W_i W_j W_k W_l] = E[W_i W_k] E[W_j W_l]$ in general, because independence of marginal distributions does not imply independence of the joint distribution. 
Assumption \ref{asmp:indep}(b) hence makes an assumption on the joint distribution. It can alternatively be stated as $(W_i, W_k) \indep (W_j, W_l)$, which is stronger but more interpretable than the zero-covariance assumption. 
I further discuss the relationship between Assumption \ref{asmp:indep} and the existing literature in Section \ref{sec:dependence}.

Assumption \ref{asmp:indep} is agnostic about the dependence structure when $W_i$ and $W_j$ share at least one cluster. It also allows the data generating process to be arbitrarily heterogeneous across different clusters, mimicking the heterogeneity permitted in one-way clustering (e.g., \citet{hansen2019asymptotic}). Since one-way clustering is a special case of two-way clustering where the $H$ cluster consists of single observations, the result here generalizes the existing results in one-way clustering. In contrast, the existing literature on two-way clustering assumes separate exchangeability that additionally imposes identical distribution over clusters, so it does not generalize the results on one-way clustering. 
For positive definite matrix $Q$, let $\lambda_{\min} (Q)$ denote the smallest eigenvalue of $Q$. Then, let $Q_n := Var \left( \sum_i W_i \right)$ denote the variance of the sum and $\lambda_n := \lambda_{\min}(Q_n)$ denote its smallest eigenvalue. For example, when $K=1$, $W_i$ is a scalar and $\lambda_n = Q_n = Var(\sum_i  W_i)$. $K_0$ is used throughout the paper to denote an arbitrary constant.

\begin{assumption} \label{asmp:regularity}
For $C \in \{G ,H \}$, and $k \in \{1, 2, \cdots, K \}$, there exists $K_0 < \infty$ such that:
\begin{enumerate}[topsep=0pt,label=(\alph*)]
    \item $E[W_{ik}^4] \leq K_0 $ for all $i$.
    \item $\frac{1}{\lambda_n} \max_c (N_c^C)^2 \rightarrow 0$.
    \item $\frac{1}{\lambda_n} \sum_c  (N_c^C)^2  \leq K_0 $.
\end{enumerate}  
\end{assumption}

Assumption \ref{asmp:regularity}(a) requires the fourth moment to be bounded, which is stronger than the moment condition in one-way clustering.\footnote{See Equation (7) of \citet{hansen2019asymptotic} for the condition in one-way clustering.} The proof in one-way clustering usually verifies a Lindeberg condition because blocks of observations are independent of each other. With two-way dependence, we no longer have independent blocks because each cluster can have observations that are dependent on observations from a different cluster when these observations share a cluster on a different dimension. Hence, a different proof strategy is required. The proof in this paper uses Stein's method, which requires stronger moment restrictions, but provides a non-asymptotic bound on the approximation error --- details are in Subsection \ref{sec:proof_sketch}. By using this strategy, a bounded fourth moment is required.

Assumption \ref{asmp:regularity}(b) requires the contribution of the largest cluster to be small relative to the total variance. This condition mimics the sparsity condition in the networks literature (e.g., \citet{graham2020sparse}). Intuitively, this condition is required so that the removal of a cluster does not change the variance substantively. This assumption allows the ratio of any two cluster sizes to diverge to infinity. It is identical to Equation (12) of \citet{hansen2019asymptotic} for one-way clustering. Assumption \ref{asmp:regularity}(b) also rules out having components that are perfectly correlated: if the components of the vector were perfectly correlated (i.e., $\mu^\prime W_i =0$ for some $\mu \ne (0,\dots,0)^\prime$), then $\lambda_n =0$. If cluster sizes are uniformly bounded, and $\lambda_n \rightarrow \infty$, then Assumption \ref{asmp:regularity}(b) is satisfied.\footnote{Assumption \ref{asmp:regularity}(b) is hence a more general version of sparsity than having the size of the dependency neighborhood (i.e., the number of observations plausibly correlated with some observation $i$) being bounded above. The conditions are also comparable with \citet{verdier2020estimation} in the two-way fixed effects literature: when the neighborhood size is bounded, $\lambda_n \asymp n$, which matches his assumption 2(c).}

Assumption \ref{asmp:regularity}(c) is a summability condition that requires $\lambda_n$ not to be too small, and requires $\lambda_n$ to be the same order as $\sum_c (N^C_c)^2$, i.e., $\lambda_n \asymp \sum_c (N^C_c)^2$, $C \in \{ G,H \}$.\footnote{For sequences $a_n$ and $b_n$, $a_n \asymp b_n$ if and only if there exists $K_0 <\infty$ such that $a_n/b_n, b_n/a_n \in [-K_0, K_0]$ for all elements in the sequence.} 
With strictly positive covariance within clusters, $\lambda_n \asymp \sum_c (N^C_c)^2$ is satisfied. However, if the researcher were conservative and clustered on $C$ when the data is indeed iid, then $\lambda_n \asymp n$, which then requires $\sum_c (N^C_c)^2 \asymp n$ for the condition to hold --- this condition holds when the cluster sizes are not too large. 
The assumption that $(1/\lambda_n) \sum_c (N_c^C)^2 \leq K_0$ matches Equation (11) of \citet{hansen2019asymptotic}.

\begin{remark}
Assumption \ref{asmp:regularity}(b) and \ref{asmp:regularity}(c) rule out the following purely interactive model. For $g \in \{ 1, \cdots, M \}$, $h \in \{ 1, \cdots, M \}$ and $N_{gh} =1$, we observe $W_{gh} = \alpha_g \gamma_h$, where $\alpha_g$ and $\gamma_h$ are iid with mean zero and variances $\sigma_\alpha^2$ and $\sigma_\gamma^2$ respectively, so there are $M^2$ observations. As pointed out by \citet{menzel2021bootstrap}, this model has an asymptotic distribution that is non-normal, with no analog in one-way clustering. To see this, $ \sum_{g,h} W_{gh}/M = \left( \sum_g  \alpha_g / \sqrt{M} \right) \left(  \sum_h \gamma_h / \sqrt{M} \right) \xrightarrow{d} Z_1 Z_2$, where $Z_1$ and $Z_2$ are independent standard normal random variables. This limiting distribution is also known as Gaussian chaos. Since $\max_g (N^G_g)^2 /\lambda_n = M^2/(M^2 \sigma_\alpha^2 \sigma_\gamma^2) = 1/(\sigma_\alpha^2 \sigma_\gamma^2)$ does not converge to 0, Assumption \ref{asmp:regularity}(b) fails. Further, $\sum_g (N^G_g)^2 /\lambda_n = M^3/(M^2 \sigma_\alpha^2 \sigma_\gamma^2) = M/\sigma_\alpha^2 \sigma_\gamma^2 \rightarrow \infty$ violates Assumption \ref{asmp:regularity}(c).
\end{remark}

\subsection{Main Result}

\begin{theorem} \label{thm:main}
Under Assumptions \ref{asmp:indep} and \ref{asmp:regularity}, $Q_n^{-1/2} \sum_i (W_i - E[W_i]) \xrightarrow{d} N(0, I_K)$. Further, if $E[W_i] = 0$ $\forall i$, then $Q_n^{-1/2} \hat{Q}_n Q_n^{-1/2} \xrightarrow{p} I_K$, where $\hat{Q}_n := \sum_i \sum_{j \in \mathcal{N}_i} W_i W_j'$.
\end{theorem}

The theorem tells us that, under the aforementioned conditions, $Q_n^{-1/2} \sum_i (W_i - E[W_i])$ is asymptotically standard normal and the plug-in variance estimator following CGM is consistent for two-way clustering. One-way clustering is a special case of this theorem when one dimension is weakly nested within the other: examples include $G=H$ so both dimensions are identical, and clustering by county and state (as counties are nested in states). A sufficient condition for consistent variance estimation is $E[W_i] = 0$, similar to Theorem 3 of \citet{hansen2019asymptotic}. This assumption is sufficient in many applications: for example, linear regressions considered in Section \ref{sec:application} are identified by requiring the expectation of the residual term to be zero. 


\begin{remark}
While the CGM variance estimator is valid in this environment without separate exchangeability, we must be more careful with bootstrap methods that were developed under separate exchangeability (e.g., \citet{menzel2021bootstrap}, \citet{mackinnon2021wild}). Bootstrap methods often resample cluster-specific means, such as $\hat{\alpha}_g = (1/N^G_g) \sum_{i \in \mathcal{N}^G_g} W_i - (1/n) \sum_i W_i$. Consider a data-generating process where, with  $\alpha_g = (1/N^G_g) \sum_{i \in \mathcal{N}^G_g} [W_i] - (1/n) \sum_i E[W_i]$, odd-numbered $g$ clusters have $\alpha_g=-1$ and even-numbered $g$ clusters have $\alpha_g=2$, and there are twice as many units in odd-numbered clusters as even-numbered clusters.
Such a process is not exchangeable. Resampling $\hat{\alpha}_g$'s with equal probability results in a positive mean, which invalidates naive bootstrap procedures. 
\end{remark}

The following two subsections discuss technicalities on the dependence structure and the proof sketch. A general-interest audience may wish to proceed immediately to Section \ref{sec:application}.

\subsection{Discussion of Dependence Structure} \label{sec:dependence}

To compare the setup used in Assumption \ref{asmp:indep} to the existing literature, I carefully define a few terms used in \citet{menzel2021bootstrap}, whose setup uses a \emph{dissociated separately exchangeable} array. Let $Y_{gh}$ denote an infinite array of observations in cluster $g$ on the $G$ dimension and cluster $h$ on the $H$ dimension.
$Y_{gh}$ is a \emph{separately exchangeable} array if, for any integers $\tilde{G}, \tilde{H}$ and permutations $\pi_1: \{ 1,\dots,\tilde{G} \} \rightarrow \{ 1,\dots,\tilde{G} \}$ and $\pi_2:  \{ 1,\dots,\tilde{H} \} \rightarrow \{ 1,\dots,\tilde{H} \}$, we have:
\begin{align*}
    (Y_{\pi_1(g)\pi_2(h)})_{g,h} \stackrel{d}{=} (Y_{gh})_{g,h}
\end{align*}
where $\stackrel{d}{=}$ denotes equality in distribution. Such an array is \emph{dissociated} if, for any $G_0, H_0 \geq 1$, $(Y_{gh})_{g=1,h=1}^{g=G_0,h=H_0}$ is independent of $(Y_{gh})_{g>G_0,h>H_0}$. Dissociation is how the existing literature formally incorporates the multi-way clustering structure. Separate exchangeability implies that the cluster indices are not meaningful, and it is stronger than having identical distributions across clusters. This environment is a special case of Assumption \ref{asmp:indep}, as the following proposition claims.
\begin{proposition} \label{prop:asmp1}
A dissociated separately exchangeable array satisfies Assumption \ref{asmp:indep}.
\end{proposition}


Assumption \ref{asmp:indep} on independence and zero covariance is also arguably more transparent and interpretable for an economics audience than formal generalizations of separate exchangeability, e.g., relative exchangeability in \citet{crane2018relatively} that requires defining a \emph{signature}, \emph{arity}, \emph{structure}, and \emph{embedding}. 

\subsection{Proof Sketch} \label{sec:proof_sketch}
The proof of Theorem \ref{thm:main} proceeds by first proving a CLT for a scalar random variable, then applying the Cramer-Wold device to obtain the multivariate CLT. The scalar CLT is proven using Stein's method. I adapt the proof strategy from \citet{ross2011fundamentals} Theorem 3.6 to obtain an upper bound on the Wasserstein distance between a pivotal statistic and the standard normal random variable. By exploiting the two-way clustering structure, the upper bound on the distance can be shown to converge to zero. All details are in Appendix \ref{sec:proof_thm1}.

For ease of exposition, consider a simpler environment where $K=1$, and $E[W_i]=0$. Let $\sigma_n^2 := Q_n$, $R = \sum_i W_i/ \sigma_n$, and $Z \sim N(0,1)$. Lemma \ref{lem:modified_ross} in Appendix \ref{sec:proof_thm1} provides an explicit bound on the Wasserstein distance between $R$ and $Z$. With $d_W(.)$ denoting the Wasserstein distance, and $d_K(.)$ denoting the Kolmogorov distance, Proposition 1.2 from \citet{ross2011fundamentals} implies that $d_K(R,Z) \leq (2/\pi)^{1/4} \sqrt{d_W(R,Z)}$. The Kolmogorov distance is the maximal distance between two CDF's, so it is informative of the maximum distance between the distribution of the pivotal statistic and the standard normal. Then, by using Assumption \ref{asmp:indep} to adapt the proof of Theorem 3.6 in \citet{ross2011fundamentals},
\begin{align*}
    d_W(R,Z) \leq  \frac{1}{\sigma_n^3} \sum_i \sum_{j,k \in \mathcal{N}_i} E[|W_i| W_j W_k] + \frac{\sqrt{2}}{\sqrt{\pi} \sigma_n^2} \sqrt{Var\left( \sum_i \sum_{j \in \mathcal{N}_i} W_i W_j \right)}
\end{align*}

Observe that this intermediate result is informative of the quality of the normal approximation. This bound on the Wasserstein distance (and hence the Kolmogorov distance) is non-asymptotic, and of the Berry-Esseen type, thereby giving a worst-case bound on the distance between the pivotal statistic and the standard normal.

At this point, my proof departs from the proofs in the existing statistical literature that employ Stein's method (e.g., \citet{chen2004normal}). Let $N_i := |\mathcal{N}_i|$. H\"{o}lder's inequality is employed on objects such as $\sum_i \sum_{j,k \in \mathcal{N}_i} E[|W_i| W_j W_k]$. The existing literature uses the $L^1$ norm of moments $E[|W_i|^3]$ and the $L^\infty$ norm of $N_i$, resulting in $\left( \max_m N_m \right)^2 \sum_i E[|W_i|^3]$. In contrast, my proof uses the $L^\infty$ norm of $E[|W_i|^3]$ and the $L^1$ norm of $N_i$, resulting in $\max_m E[|W_m|^3] \sum_i N_i^2$. Hence,
\begin{align*}
    \frac{1}{\sigma_n^3} \sum_i \sum_{j,k \in \mathcal{N}_i} E[|W_i| W_j W_k]  \leq \frac{1}{\sigma_n^3} \max_m E[|W_m|^3] \sum_i N_i^2 
\end{align*}

Since $\max_m E[|W_m|^3]$ is bounded by Assumption \ref{asmp:regularity}(a), it suffices to show that $\sum_i N_i^2 / \sigma_n^3 \rightarrow 0$. Due to Assumption \ref{asmp:indep}(a), $N_i \leq N^G_{g(i)} + N^H_{h(i)}$, so
\begin{align*}
    \frac{1}{\sigma_n^3} \sum_i N_i^2 \leq \frac{1}{\sigma_n^3} \sum_i (N^G_{g(i)} + N^H_{h(i)})^2 & \leq \frac{1}{\sigma_n^3} \max_{g,h} (N^G_g + N^H_h) \sum_i (N_{g(i)} + N_{h(i)}) \\
    &\leq \left[ \frac{1}{\sigma_n} \max_{g,h} (N^G_g + N^H_h) \right] \frac{1}{\sigma_n^2}\left( \sum_g (N^G_g)^2 + \sum_h (N^H_h)^2 \right)
\end{align*}
Since $\lambda_n = \sigma_n^2$ when $K=1$, $\max_{g,h} (N^G_g + N^H_h) /\sigma_n \rightarrow 0$ by Assumption \ref{asmp:regularity}(b) and the final term $\left( \sum_g (N^G_g)^2 + \sum_h (N^H_h)^2 \right) /\sigma_n^2$ is bounded by Assumption \ref{asmp:regularity}(c). Hence, the term is $o(1)$.

A similar argument is made for the fourth moment that features in $Var\left( \sum_i \sum_{j \in \mathcal{N}_i} W_i W_j \right)$. To complete the proof for variance estimation, observe that since the fourth moments exist, the consistency of the plug-in variance estimator can be proven by using Chebyshev's inequality and the existing intermediate results. 

\begin{remark}
By modifying the proof of Theorem 3.6 in \citet{ross2011fundamentals}, the conditions in this paper permit some forms of heterogeneity in cluster sizes that Theorem 3.6 of \citet{ross2011fundamentals} does not. The following is one such example. All observations are the only observation in their $H$ cluster, i.e., $h(i)=i$. On the $G$ dimension, the first cluster has size $N^G_1 = n^{1/4}$, while all other clusters have size 1. Then, $\lambda_n \asymp n^{1/2} + (n - n^{1/4}) \asymp n$ and $(N^G_1)^2 / \lambda_n \asymp n^{1/2}/n = o(1)$, so the conditions of Theorem \ref{thm:main} are satisfied. However, Theorem 3.6 of \citet{ross2011fundamentals} bounds the Wasserstein distance by $\left( N_1^{2}/ \lambda_n^{3/2} \right) \sum_i E |W_i|^3$ and a term that involves the fourth moment. We have $N_1^{2}/ \lambda_n^{3/2} \sum_i E |W_i|^3 \asymp n^{-1} \sum_i E |W_i|^3 \ne o(1)$, so we may not obtain convergence.  
\end{remark}

\begin{remark}
There are several early papers in the probability theory literature that deliver similar results, but are insufficient for Theorem \ref{thm:main}. For instance, Theorem 2 of \citet{janson1988normal} is a central limit theorem that uses the condition (with $m=3$):
\begin{align*}
    \left( \frac{n}{\max_i N_i} \right)^{1/3} \frac{\left( \max_i N_i \right) \max_i |W_i|}{\sigma_n} =  \left( \frac{n}{\sigma_n^3} \left( \max_i N_i \right)^2 \right)^{1/3}  \max_i |W_i| \rightarrow 0
\end{align*}
In this proof sketch, I have shown that $\sum_i N_i^2 / \sigma_n^3 \rightarrow 0$, but $n (\max_i N_i )^2 / \sigma_n^3 \geq \sum_i N_i^2 / \sigma_n^3$, so the \citet{janson1988normal} condition need not hold in this environment.
\end{remark}

\section{Theory for Least Squares Regression} \label{sec:application}
This section applies Theorem \ref{thm:main} to linear regressions, showing that using the normal approximation with the CGM variance estimator is valid. Consider a linear model where the scalar outcome $Y_i$ is generated by
\begin{align*}
    Y_i = X_i'\beta + u_i
\end{align*}

with $X_i \in \mathbb{R}^K$. We are interested in estimating $\beta$. Suppose $E[X_i u_i] =0$ for all $i$, and $(X_i^\prime, u_i)$ is allowed to be two-way clustered. The standard OLS estimator is
\begin{align*}
    \hat{\beta} = \left(  \sum_i X_i X_i' \right)^{-1} \left( \sum_i X_i Y_i \right) = \beta + \left(  \sum_i X_i X_i' \right)^{-1} \left( \sum_i X_i u_i \right)
\end{align*}

This object is assumed to be well-defined in that $\sum_i X_i X_i'$ is invertible. Define $S_n := \sum_i E[X_i X_i']$ and $Q_n := Var \left( \sum_i X_i u_i \right)$, and denote their sample analogs as $\hat{S}_n = \sum_i X_i X_i'$ and $\hat{Q}_n := \sum_i \sum_{j \in \mathcal{N}_i} \hat{u}_i \hat{u}_j X_i X_j'$. Let the smallest eigenvalue of $Q_n$ be $\lambda_n := \lambda_{\min}(Q_n)$. The asymptotic variance of $\hat{\beta}$ and its sample analog are $V(\hat{\beta}) := S_n^{-1} Q_n S_n^{-1}$ and $ \hat{V}(\hat{\beta}) := \hat{S}_n^{-1} \hat{Q}_n \hat{S}_n^{-1}$ respectively.

Assumption \ref{asmp:stochasticX} provides sufficient conditions for the estimator $\hat{\beta}$ to be asymptotically normal and for the CGM variance estimator to be consistent. The conditions mimic Assumption \ref{asmp:regularity} so that Theorem \ref{thm:main} is applicable to the random vector $X_i u_i$. The new condition is a weak regularity condition that $\lambda_{\min} \left( S_n /n \right) \geq K_1 > 0$, mimicking the rank condition in OLS.

\begin{assumption} \label{asmp:stochasticX}
For $C \in \{G ,H \}$, and $k \in \{1, 2, \cdots, K \}$, there exists $K_0<\infty$ and $K_1 >0$ such that:
\begin{enumerate}[topsep=0pt,label=(\alph*)]
    \item $E[u_i^4| X_i] \leq K_0$, $E[X_{ik}^4] \leq K_0$, $E[X_i u_i] =0$ for all $i$.
    \item $\frac{1}{\lambda_n} \max_c (N^C_c)^2 \rightarrow 0$.
    \item $\frac{1}{\lambda_n} \sum_c (N^C_c)^2 \leq K_0$.
    \item $(X_i', u_i)' \indep \{(X_j', u_j)'\}_{j \notin \mathcal{N}_i}$. For observations $i,j $ and $k \in \mathcal{N}_i,l \in \mathcal{N}_j$ and all nonstochastic $\mu \in \mathbb{R}^K$, if $j,l \notin (\mathcal{N}_i \cup \mathcal{N}_k)$, then $(X_i', u_i, X_k', u_k)' \indep (X_j', u_j, X_l', u_l)'$.
    \item $\lambda_{\min} \left( \frac{1}{n} S_n \right) \geq K_1 $. 
\end{enumerate}  
\end{assumption}

\begin{proposition} \label{prop:stochasticX}
Under Assumption \ref{asmp:stochasticX}, $Q_n^{-1/2} S_n (\hat{\beta} - \beta) \xrightarrow{d} N(0,I_K)$, and $[S_n^{-1} Q_n S_n^{-1}]^{-1} [\hat{S}_n^{-1} \hat{Q}_n \hat{S}_n^{-1}] \xrightarrow{p} I_K$.
\end{proposition}

Proposition \ref{prop:stochasticX} is useful for performing F tests on a subvector of $\beta$. The proof of Proposition \ref{prop:stochasticX} proceeds by applying Theorem \ref{thm:main} to $\sum_i X_i u_i$, then showing that $S_n^{-1} \hat{S}_n \xrightarrow{p} I_K$, which uses the rank condition of Assumption \ref{asmp:stochasticX}(e). It then remains to show that the remainder terms are asymptotically negligible. 

The practitioner's takeaway from Proposition \ref{prop:stochasticX} is that the existing CGM variance estimator can be used for valid inference with two-way clustering. 
The result provides the formal theoretical guarantee for using the estimator, under conditions that permit heterogeneity across clusters. 

Besides the application mentioned, Theorem \ref{thm:main} also has implications on the conditions required for valid inference when the random variable is two-way clustered in many other econometric models, including design-based settings and instrumental variables models. This theory is especially relevant for design-based settings where the researcher conditions on potential outcomes, so the random variable cannot be separately exchangeable by construction --- see \citet{yap2023design}, for instance. Inference for estimators based on moment conditions can be done by straightforward application of Theorem \ref{thm:main} as in linear regression. Practically, this paper has shown that the popular CGM estimator is robust in an environment without separate exchangeability, but practitioners should exercise caution when applying bootstrap methods to environments that are not separately exchangeable. While the results are presented for two-way clustering, they can be easily extended to clustering on three or more dimensions.

\appendix
\renewcommand{\thetheorem}{\thesection.\arabic{theorem}}
\setcounter{theorem}{0}
\section{Proof of Theorem 1} \label{sec:proof_thm1}
The proof strategy is as follows. I first prove Lemma \ref{lem:scalar_CLT}, which is a central limit theorem (CLT) for scalars. The proof of Lemma \ref{lem:scalar_CLT} relies on Lemmas \ref{lem:ross3.1} to \ref{lem:4th_moment}. Lemmas 2 to 4 derive an upper bound on the Wasserstein distance between a pivotal statistic and standard normal $Z$. Lemmas 5 to 7 then show that the derived upper bound is $o(1)$. With Lemma \ref{lem:scalar_CLT}, the multivariate CLT of Theorem \ref{thm:main} is obtained by using the Cramer-Wold device. The remainder of the proof proceeds in the following order: (i) introduce definitions and notation, (ii) state Lemma \ref{lem:scalar_CLT}, (iii) state and prove Lemmas \ref{lem:ross3.1} to \ref{lem:4th_moment}, (iv) prove Lemma \ref{lem:scalar_CLT}, then (v) complete the proof of Theorem \ref{thm:main}. 


The following definitions and notations are used throughout the proof. Let $d_W(X,Y)$ denote the Wasserstein distance between random variables $X$ and $Y$, so $d_W(X,Y) =0$ if and only if the distributions of $X$ and $Y$ are identical. The norms of functions are defined as the sup norm i.e., $||f|| = \sup_{x \in D} |f(x)|$. For vector $a$, $||a|| = (a'a)^{1/2}$ is the Euclidean norm, and for positive semi-definite matrix $A$ and $\lambda_{\max}(A)$ denoting the largest eigenvalue, $||A|| = \sqrt{\lambda_{\max}(A'A)}$ denotes the spectral norm, and $A^{1/2}$ denotes the symmetric matrix such that $A^{1/2} A^{1/2} = A$. $\sum_{i \in \mathcal{N}^G_g} \sum_{j \in \mathcal{N}^G_g}$ is abbreviated as $\sum_{i,j \in \mathcal{N}^G_g}$. The dependency neighborhood of $i$, $\mathcal{N}_i \subseteq \{ 1, \cdots, n \}$, is defined as the set of observations where $i \in \mathcal{N}_i$ and $X_i$ is independent of $\{ X_j \}_{j \ne \mathcal{N}_i}$, and $N_i := |\mathcal{N}_i|$ is the number of observations in $i$'s dependency neighborhood. $1[A]$ is an indicator function that takes value 1 if $A$ is true and 0 otherwise. In the rest of this proof, $X_i$ denotes a scalar random variable while $W_i \in \mathbb{R}^K$ as stated in the main text is a random vector. Denote the variance of the sum of the scalar random variable $X_i$ as $\sigma_n^2 := Var \left( \sum_i X_i \right)$. We are interested in the asymptotic distribution of $(1/\sigma_n) \sum_i X_i$. 

\begin{assumption} \label{asmp:scalar_CLT}
For $C \in \{G ,H \}$, there exists $K_0 < \infty$ such that:
\begin{enumerate}[label=(\alph*)]
    \item $E[X_i]=0$ and $E[X_i^4] \leq K_0 <\infty$ for all $i$.
    \item $\frac{1}{\sigma_n^2} \max_c \left( N^C_c \right)^2 \rightarrow 0$
    \item $\frac{1}{\sigma_n^2} \sum_c \left( N^C_c \right)^2  \leq K_0 < \infty$
    \item $X_i \indep \{ X_j \}_{j \notin \mathcal{N}_i}$.
    \item For observations $i,j,k \in \mathcal{N}_i,l \in \mathcal{N}_j$, if $(\mathcal{N}_i \cup \mathcal{N}_k) \cap (\mathcal{N}_j \cup \mathcal{N}_l) = \emptyset$, then $Cov(X_i X_k, X_j X_l) =0$.
\end{enumerate}  
\end{assumption}

\begin{lemma} \label{lem:scalar_CLT}
Under Assumption \ref{asmp:scalar_CLT}, $(1/\sigma_n) \sum_i X_i \xrightarrow{d} N(0,1)$, where $\sigma_n^2 := Var\left( \sum_i X_i \right)$. Further, using feasible estimator $\hat{\sigma}_n^2 := \sum_i \sum_{j \in \mathcal{N}_i}  X_i X_j$, $\hat{\sigma}_n^2/ \sigma_n^2 \xrightarrow{p} 1$.
\end{lemma}

\begin{lemma} \label{lem:ross3.1}
(Theorem 3.1 of \citet{ross2011fundamentals}) If $R$ is a random variable, $Z$ has a standard normal distribution, and we define the family of functions $\mathcal{F} = \{ f: ||f||, ||f^{\prime \prime}|| \leq 2,  ||f^{\prime}|| \leq \sqrt{2\pi} \}$, then $d_W(R,Z) \leq \sup_{f \in \mathcal{F}} |E[f'(R) - Rf(R)]|$.
\end{lemma}


The proofs of Lemmas \ref{lem:ross3.3} and \ref{lem:modified_ross} follow \citet{ross2011fundamentals} Theorem 3.6 up to Equations (3.11) and (3.12).

\begin{lemma} \label{lem:ross3.3}
Let $X_1, \cdots, X_n$ be random variables such that $E[X_i] =0, \sigma_n^2 = Var(\sum_i X_i)$, and define $R = \sum_i X_i/ \sigma_n$. If $R_i := \sum_{j \notin \mathcal{N}_i} X_j/ \sigma_n$, then, for all $f \in \mathcal{F}$, 
\begin{align*}
    E[Rf(R)] &= E\left[ \frac{1}{\sigma_n} \sum_i X_i (f(R) - f(R_i) - (R-R_i) f'(R)) \right] + E\left[ \frac{1}{\sigma_n} \sum_i X_i (R-R_i) f'(R) \right]
\end{align*}
\end{lemma}

\begin{proof}
Start from right-hand side: 
\begin{align*}
    E& \left[ \frac{1}{\sigma_n} \sum_i X_i (f(R) - f(R_i) - (R-R_i) f'(R)) \right] + E\left[ \frac{1}{\sigma_n} \sum_i X_i (R-R_i) f'(R) \right] \\
    &= E \left[ \frac{1}{\sigma_n} \sum_i X_i (f(R) - f(R_i)) \right] = E \left[ \frac{1}{\sigma_n} \sum_i X_i f(R)  \right] + E \left[ \frac{1}{\sigma_n} \sum_i X_i f(R_i) \right] \\
    &= E \left[ \frac{1}{\sigma_n} \sum_i X_i f(R)  \right] =E[Rf(R)]
\end{align*}
The first equality in the final line comes from the fact that $R_i$ is independent of $X_i$ based on how dependency neighborhoods are defined. Hence, $E[X_i f(R_i)]=0$. 
\end{proof}

\begin{lemma} \label{lem:modified_ross}
Let $X_1, \cdots, X_n$ be random variables such that, $E[X_i] =0, \sigma_n^2 = Var(\sum_i X_i)$, and define $R = \sum_i X_i/ \sigma_n$. Let the collection $(X_1, \cdots, X_n)$ have dependency neighborhoods $\mathcal{N}_i$, $i = 1, \cdots, n$. Then for $Z$ a standard normal random variable,
\begin{equation} \label{eqn:dW_bound}
    d_W(R,Z) \leq \frac{1}{\sigma_n^3} \sum_i \sum_{j,k \in \mathcal{N}_i} E \left[  |X_i| X_j X_k  \right] + \frac{\sqrt{2}}{\sqrt{\pi} \sigma_n^2} \sqrt{Var\left( \sum_i \sum_{j \in \mathcal{N}_i} X_i X_j \right)}
\end{equation}
\end{lemma}

\begin{proof}
Due to Lemma \ref{lem:ross3.1}, to bound $d_W(R,Z)$ from above, it is sufficient to bound $|E[f'(R) - Rf(R)]|$, where $||f||, ||f^{\prime \prime}|| \leq 2,  ||f^{\prime}|| \leq \sqrt{2/ \pi}$. Define $R_i := \sum_{j \notin \mathcal{N}_i} X_j/ \sigma_n$, so $X_i$ is independent of $R_i$. 
\begin{align*}
    |E[&f'(R) - Rf(R)]| = |E[f'(R)] - E[Rf(R)]| \\
    &\leq \left|E[f'(R)] - E\left[ \frac{1}{\sigma_n} \sum_i X_i (f(R) - f(R_i) - (R-R_i) f'(R)) \right] - E\left[ \frac{1}{\sigma_n} \sum_i X_i (R-R_i) f'(R) \right] \right| \\
    &\leq \left| E\left[ \frac{1}{\sigma_n} \sum_i X_i(f(R) - f(R_i) - (R-R_i) f'(R))\right] \right| + \left| E\left[ f'(R) \left( 1 - \frac{1}{\sigma_n} \sum_i X_i(R-R_i) \right) \right] \right|
\end{align*}

The first inequality applies Lemma \ref{lem:ross3.3}, and the second inequality applies the triangle inequality. Consequently, it is sufficient to show that the first term is bounded by the corresponding first term of Equation (\ref{eqn:dW_bound}), and the second term is bounded by the corresponding second term.

Consider the first term.  By Taylor expansion of $f(R_i)$ around $f(R)$, and the triangle inequality, the term that generates the third moment is:
\begin{align*}
    \Bigg| E &\left[ \frac{1}{\sigma_n} \sum_i X_i(f(R) - f(R_i) - (R-R_i) f'(R))\right] \Bigg| \leq \frac{||f^{\prime \prime}||}{2 \sigma_n} \left| \sum_i E [|X_i| (R- R_i)^2] \right| \\
    &= \frac{1}{\sigma_n^3} \sum_i E\left[|X_i| \left( \sum_{j \in \mathcal{N}_i} X_j \right)^2 \right]= \frac{1}{\sigma_n^3} \sum_i \sum_{j,k \in \mathcal{N}_i} E[|X_i| X_j X_k] 
\end{align*} 
Turning now to the second term,
\begin{align*}
    & \left| E\left[ f'(R) \left( 1 - \frac{1}{\sigma_n} \sum_i X_i (R-R_i) \right) \right] \right|  \\
    &\leq \frac{||f'||}{\sigma_n^2} E \left| \sigma_n^2 - \sum_i X_i \left( \sum_{j \in N_i} X_j \right)  \right| \leq \frac{||f'||}{\sigma_n^2} E\left[ \left(  \sigma_n^2 - \sum_i X_i \left( \sum_{j \in N_i} X_j \right) \right)^2 \right]^{1/2} 1^{1/2} \\
    &\leq \frac{\sqrt{2}}{\sqrt{\pi} \sigma_n^2} \sqrt{Var\left( \sum_i \sum_{j \in N_i} X_i X_j \right)}
\end{align*}

\end{proof}

\begin{lemma} \label{lem:AMGM}
$E[|X_i X_j X_k|] \leq \max_m E[|X_m|^3]$, $E[|X_i X_j X_k X_l|] \leq \max_m E[|X_m|^4]$, and $|E[X_i X_k] E[X_j X_l]| \leq \max_m E[|X_m|^4]$.
\end{lemma}

\begin{proof}
By the arithmetic mean --- geometric mean (AM-GM) inequality, 
$$E|X_i X_j X_k| \leq \frac{1}{3} \left( E|X_i|^3 + E|X_j|^3 + E|X_k|^3 \right) \leq \max_m E[|X_m|^3]$$
A similar argument yields $E[|X_i X_j X_k X_l|] \leq \max_m E[|X_m|^4]$. For the final result, first observe that $E[X_i X_k]^2 \pm 2 E[X_i X_k] E[X_j X_l] + E[X_j X_l]^2 = (E[X_i X_k] \pm E[X_j X_l])^2 \geq 0$. Hence,
\begin{align*}
    |E[X_i X_k] E[X_j X_l]| &\leq \frac{1}{2} (E[X_i X_k]^2 + E[X_j X_l]^2 ) \leq \frac{1}{2} (E[X_i^2 X_k^2] + E[X_j^2 X_l^2] ) \\
    &\leq \frac{1}{4} (E[X_i^4] + E[X_j^4] + E[X_k^4] + E[X_l^4]) \leq \max_m E[X_m^4]
\end{align*}
\end{proof}

\begin{lemma} \label{lem:3rd_moment}
Under Assumption \ref{asmp:scalar_CLT}, $\frac{1}{\sigma_n^3} \sum_i \sum_{j,k \in \mathcal{N}_i} E \left[  |X_i| X_j X_k  \right] \rightarrow 0$.
\end{lemma}
\begin{proof}
Using Lemma \ref{lem:AMGM},
\begin{align*}
    \frac{1}{\sigma_n^3} \sum_i \sum_{j,k \in \mathcal{N}_i} E \left[  |X_i| X_j X_k  \right] &\leq  \frac{1}{\sigma_n^3} \sum_i \sum_{j,k \in \mathcal{N}_i} E \left[  \left| |X_i| X_j X_k \right|  \right] \\
    &\leq \frac{\max_m E[|X_m|^3]}{\sigma_n^3} \sum_i  \sum_{j,k \in \mathcal{N}_i} 1 = \frac{\max_m E[|X_m|^3]}{\sigma_n^3} \sum_i N_i^2
\end{align*}

Observe $\max_m E[|X_m|^3] \leq K_0$ since the 4th moment exists, so it remains to show that the remaining terms are $o(1)$. 

Due to Assumption \ref{asmp:indep}, $N_i \leq N^G_{g(i)} + N^H_{h(i)}$, so
\begin{align*}
    \frac{1}{\sigma_n^3} \sum_i N_i^2 \leq \frac{1}{\sigma_n^3} \sum_i (N^G_{g(i)} + N^H_{h(i)})^2 & \leq \frac{1}{\sigma_n^3} \max_{g,h} (N^G_g + N^H_h) \sum_i (N_{g(i)} + N_{h(i)}) \\
    &\leq \left[ \frac{1}{\sigma_n} \max_{g,h} (N^G_g + N^H_h) \right] \frac{1}{\sigma_n^2}\left( \sum_g (N^G_g)^2 + \sum_h (N^H_h)^2 \right)
\end{align*}
$\max_{g,h} (N^G_g + N^H_h) /\sigma_n \rightarrow 0$ by Assumption \ref{asmp:regularity}(b) and the final term $\left( \sum_g (N^G_g)^2 + \sum_h (N^H_h)^2 \right) /\sigma_n^2$ is bounded by Assumption \ref{asmp:regularity}(c). Hence, the term is $o(1)$.
\end{proof}

\begin{lemma} \label{lem:4th_moment}
Under Assumption \ref{asmp:scalar_CLT}, $\frac{1}{\sigma_n^4} Var\left( \sum_i \sum_{j \in \mathcal{N}_i} X_i X_j \right) = o(1)$.
\end{lemma}

\begin{proof}

\begin{align*}
    \frac{1}{\sigma_n^4} &Var \left( \sum_i \sum_{j \in \mathcal{N}_i}  X_i X_j \right) =\frac{1}{\sigma_n^4} E \left[ \left( \sum_i \sum_{j \in \mathcal{N}_i}  X_i X_j \right)^2 \right] - \frac{1}{\sigma_n^4} \left( \sum_i \sum_{j \in \mathcal{N}_i} E[ X_i X_j] \right)^2 \\
    &= \frac{1}{\sigma_n^4} \sum_i \sum_j \sum_{k \in \mathcal{N}_i} \sum_{l \in \mathcal{N}_j} (E[  X_i X_j X_k X_l] - E[ X_i X_k] E[X_j X_l])
\end{align*}

Due to Assumption \ref{asmp:indep}(b), when $j,l$ do not share any cluster with $i,k$, $E[X_i X_j X_k X_l] = E[X_i X_k] E [X_j X_l]$. Hence, we only have to consider terms where there is at least one pair that shares a cluster. Let $A_{ij} := 1[j \in \mathcal{N}_i ]$. With finite 4th moment and Lemma \ref{lem:AMGM}, using the same argument as the proof of Lemma \ref{lem:3rd_moment}, it is sufficient to show
\begin{align*}
    \frac{1}{\sigma_n^4} \sum_i \sum_{j} \sum_{k \in \mathcal{N}_i} \sum_{l \in \mathcal{N}_j}  (A_{ij} + A_{il} + A_{kj} + A_{kl}) = o(1)
\end{align*}
It is sufficient to consider the $A_{ij}$ term because everything else is analogous. 
\begin{align*}
    \sum_i \sum_{j} \sum_{k \in \mathcal{N}_i} \sum_{l \in \mathcal{N}_j} A_{ij} \leq \sum_i \left(  \sum_{j \in \mathcal{N}^G_{g(i)}} + \sum_{j \in \mathcal{N}^H_{h(i)}} \right) \left(  \sum_{k \in \mathcal{N}^G_{g(i)}} + \sum_{k \in \mathcal{N}^H_{h(i)}} \right) \left(  \sum_{l \in \mathcal{N}^G_{g(j)}} + \sum_{l \in \mathcal{N}^H_{h(j)}} \right) A_{ij}
\end{align*}

The first and last terms of the summation take the form: 
\begin{align*}
    \sum_i \sum_{j \in \mathcal{N}^G_{g(i)}} \sum_{k \in \mathcal{N}^G_{g(i)}} \sum_{l \in \mathcal{N}^G_{g(j)}} A_{ij} &= \sum_g \sum_{i,j,k,l \in \mathcal{N}^G_g} A_{ij} \leq \left( \max_g \sum_{k,l \in \mathcal{N}^G_g} 1\right) \sum_g \sum_{i,j \in \mathcal{N}^G_g}  A_{ij}
\end{align*}

Since $\frac{1}{\sigma_n^2} \max_g \sum_{k,l \in \mathcal{N}^H_h} 1 = \frac{1}{\sigma_n^2} \max_g \left( N^G_g \right)^2  = o(1)$ and $\frac{1}{\sigma_n^2}  \sum_g \sum_{i,j \in \mathcal{N}^G_g}  A_{ij} \leq \frac{1}{\sigma_n^2}  \sum_g (N^G_g)^2 < \infty$ by Assumption \ref{asmp:scalar_CLT}, these terms are $o(1)$ when divided by $\sigma_n^4$. The interactive terms have the form:
\begin{align*}
    \sum_i &\sum_{j \in \mathcal{N}^G_{g(i)}} \sum_{k \in \mathcal{N}^G_{g(i)}} \sum_{l \in \mathcal{N}^H_{h(j)}} A_{ij} \\
    &= \sum_{i,j,k} \sum_g 1[i \in \mathcal{N}^G_g] 1[j \in \mathcal{N}^G_g] 1[k \in \mathcal{N}^G_g] \sum_l \sum_h 1[j \in \mathcal{N}^H_h] 1[l \in \mathcal{N}^H_h] A_{ij} \\
    &= \sum_j \sum_{i,k} \sum_g 1[i \in \mathcal{N}^G_g] 1[j \in \mathcal{N}^G_g] 1[k \in \mathcal{N}^G_g]  A_{ij} \sum_h \sum_l 1[j \in \mathcal{N}^H_h] 1[l \in \mathcal{N}^H_h]  \\
    &\leq \left( \max_j \sum_h \sum_l 1[j \in \mathcal{N}^H_h] 1[l \in \mathcal{N}^H_h]  \right) \left( \sum_g \sum_{i,j,k \in \mathcal{N}^G_g}  A_{ij} \right) \\
    &= \left( \max_h \sum_{l \in \mathcal{N}^H_h} 1 \right) \left( \max_g \sum_{k \in \mathcal{N}^G_g} 1\right) \left( \sum_g \sum_{i,j \in \mathcal{N}^G_g}  A_{ij} \right) =  \left( \max_h N^H_h \right) \left( \max_g  N^G_g \right) \left( \sum_g (N^G_g)^2 \right) 
\end{align*}

Since $ \sum_g (N^G_g)^2 /\sigma_n^2  \leq K_0$ and $ \max_g N^G_g / \sigma_n  = o(1)$,
\begin{align*}
    \frac{1}{\sigma_n^4} &\sum_i \sum_{j \in \mathcal{N}^G_{g(i)}} \sum_{k \in \mathcal{N}^G_{g(i)}} \sum_{l \in \mathcal{N}^H_{h(j)}} A_{ij} \leq \left( \frac{1}{\sigma_n} \max_h N^H_h \right) \left( \frac{1}{\sigma_n} \max_g N^G_g \right) \left( \frac{1}{\sigma_n^2} \sum_g  (N^G_g)^2 \right) = o(1)
\end{align*}

\end{proof}

\begin{proof} [Proof of Lemma \ref{lem:scalar_CLT}]
Apply Lemma \ref{lem:modified_ross} to obtain:
\begin{align*}
    d_W(R,Z) \leq \frac{1}{\sigma_n^3} \sum_i  \sum_{j,k \in \mathcal{N}_i} E[ |X_i| X_j X_k]  + \frac{\sqrt{2}}{\sqrt{\pi} \sigma_n^2} \sqrt{Var\left( \sum_i \sum_{j \in \mathcal{N}_i}  X_i X_j \right)}
\end{align*}
Applying Lemma \ref{lem:3rd_moment} and \ref{lem:4th_moment} on each of the two terms, $d_W(R,Z) = o(1)$. Proof for consistency of the variance estimator is equivalent to proving that $(\hat{\sigma}_n^2 - \sigma_n^2)/\sigma_n^2 = o_P(1)$. By Chebyshev's inequality, 
\begin{align*}
    P\left( \frac{\hat{\sigma}_n^2 - \sigma_n^2}{\sigma_n^2} > \epsilon \right) &\leq \frac{1}{\epsilon^2} \frac{1}{\sigma_n^4} E [(\hat{\sigma}_n^2 - \sigma_n^2)^2] = \frac{Var \left( \sum_i \sum_{j \in \mathcal{N}_i}  X_i X_j \right)}{\epsilon^2 \sigma_n^4} = o_P(1)
\end{align*}

The convergence in the last step occurs by Lemma \ref{lem:4th_moment}. 
\end{proof}
\begin{proof} [Proof of Theorem \ref{thm:main}]

To show that $Q_n^{-1/2} \sum_i (W_i - E[W_i]) \xrightarrow{d} N(0, I_K)$, due to the Cramer-Wold device, it suffices to show that $\forall \mu \in \mathbb{R}^K$, $\mu'Q_n^{-1/2} \sum_i(W_i - E[W_i]) \xrightarrow{d} \mu'N(0, I_K)$. If $\mu$ is a vector of zeroes, then $\mu'Q_n^{-1/2} \sum_i (W_i - E[W_i]) \xrightarrow{d} \mu'N(0, I_K)$ is immediate. For $||\mu||>0$, it suffices to show $(1/||\mu||)\mu'Q_n^{-1/2} \sum_i (W_i - E[W_i]) \xrightarrow{d} (1/||\mu||) \mu'N(0, I_K) = N(0,1)$. Without loss of generality, we can set $||\mu||=1$. For all nonstochastic $\mu \in \mathbb{R}^K \backslash \{ 0 \}$, let $\sigma_n^2(\mu) := Var \left( \sum_i \mu' \left( Q_n/ \lambda_n \right)^{-1/2} (W_i - E[W_i]) \right)$, so the following hold:
\begin{enumerate}
    \item $E \left[ \left( \mu' \left( \frac{1}{\lambda_n} Q_n \right)^{-1/2} (W_i - E[W_i]) \right) \right]=0$ and $E\left[\left( \mu' \left( \frac{1}{\lambda_n} Q_n \right)^{-1/2} (W_i - E[W_i]) \right)^4 \right] \leq K_0$ for all $i$.
    \item $\frac{1}{\sigma_n^2(\mu)} \max_c \left( N^C_c \right)^2 \rightarrow 0$.
    \item $\frac{1}{\sigma_n^2(\mu)} \sum_c  (N^C_c)^2  \leq K_0$.
    \item $\left( \mu' \left( \frac{1}{\lambda_n} Q_n \right)^{-1/2} (W_i - E[W_i]) \right) \indep \left\{ \left( \mu' \left( \frac{1}{\lambda_n} Q_n \right)^{-1/2} W_j \right) \right\}_{j \notin \mathcal{N}_i}$.
    \item For observations $i,j,k \in \mathcal{N}_i,l \in \mathcal{N}_j$, if $(\mathcal{N}_i \cup \mathcal{N}_k) \cap (\mathcal{N}_j \cup \mathcal{N}_l) = \emptyset$, then 
    $$Cov\left( \mu' \left( \frac{1}{\lambda_n} Q_n \right)^{-1/2} W_i  \mu' \left( \frac{1}{\lambda_n} Q_n \right)^{-1/2} W_k, \mu' \left( \frac{1}{\lambda_n} Q_n \right)^{-1/2} W_j \mu' \left( \frac{1}{\lambda_n} Q_n \right)^{-1/2} W_l  \right) =0.$$
\end{enumerate}

For item 1, since $\lambda_n := \lambda_{\min} (Q_n)$, all eigenvalues of $Q_n/\lambda_n$ must be at least 1. Hence, all eigenvalues of $(Q_n/ \lambda_n)^{-1/2}$ are bounded above by 1, which implies $|\mu'(Q_n/ \lambda_n)^{-1/2}| \leq K_1$ for some arbitrary constant $K_1 <\infty$. Item 1 then follows from Assumption \ref{asmp:regularity}(a). Observe that $\sigma_n^2(\mu) = \mu' (Q_n/ \lambda_n)^{-1/2} Q_n (Q_n/ \lambda_n)^{-1/2} \mu = \lambda_n$. Then, Assumption \ref{asmp:regularity}(b) yields item 2, and Assumption \ref{asmp:regularity}(c) yields item 3. Item 4 is immediate from Assumption \ref{asmp:indep}(a), and item 5 from Assumption \ref{asmp:indep}(b). By applying Lemma \ref{lem:scalar_CLT}, $(1/\sigma_n(\mu)) \mu'(Q_n/\lambda_n)^{-1/2} \sum_i (W_i - E[W_i]) \xrightarrow{d} N(0,1)$. By using $\sigma_n^2(\mu) = \lambda_n$, this result is equivalent to $\mu'Q_n^{-1/2} \sum_i (W_i - E[W_i]) \xrightarrow{d} N(0,1)$ as required. 

Turning to consistent variance estimation, I first show that  $(1/\lambda_n) (\hat{Q}_n - Q_n) \xrightarrow{p} 0_{K \times K}$, where $0_{K \times K}$ is a $K \times K$ matrix of zeroes. Since $\hat{Q}_n - Q_n = \sum_i \sum_{j \in \mathcal{N}_i}  W_i W_j' - E[ W_i W_j']$, it suffices to show convergence elementwise. Let $X_i$ and $Y_i$ denote scalar components of $W_i$, i.e., $X_i = W_{im}, Y_i = W_{ip}$, where $m,p \in \{1, 2, \cdots, K \}$.
\begin{align*}
    P&\left( \frac{1}{\lambda_n} \sum_i \sum_{j \in \mathcal{N}_i}  (X_i Y_j - E[X_i Y_j]) > \epsilon \right) \leq \frac{1}{\epsilon^2} \frac{1}{\lambda_n^2} Var \left( \sum_i \sum_{j \in \mathcal{N}_i}  X_i Y_j \right) \\
    &\leq \frac{1}{\epsilon^2 \lambda_n^2} \sum_i \sum_j \sum_{k \in \mathcal{N}_i} \sum_{l \in \mathcal{N}_j} \left| E[ X_i X_j Y_k Y_l] - E[ X_i Y_k] E[ X_j Y_l] \right| \\
    &\leq \frac{K_0}{\lambda_n^2} \sum_i \sum_j \sum_{k \in \mathcal{N}_i} \sum_{l \in \mathcal{N}_j} (A_{ij} + A_{il} + A_{kj} + A_{kl}) = o(1)
\end{align*}

The inequality in the last line is obtained due to H\"{o}lder's inequality and finite moments. An argument similar to that of Lemma \ref{lem:4th_moment} yields the $o(1)$ equality. Then,
\begin{align*}
\mu'(Q_n^{-1/2} (\hat{Q}_n - Q_n) Q_n^{-1/2}) \mu &= \mu_0' \frac{1}{\lambda_n} (\hat{Q}_n - Q_n) \mu_0 \xrightarrow{p} 0
\end{align*}
where $\mu_0$ is a vector whose entries are all bounded above by some arbitrary constant $K_1 < \infty$ by a similar argument as before. Convergence occurs because $(1/\lambda_n) (\hat{Q}_n - Q_n) \xrightarrow{p} 0_{K \times K}$.

\end{proof}

\section{Proof of Propositions}

\begin{proof} [Proof of Proposition \ref{prop:asmp1}]
For (a), take any observation $i$ and its associated clusters $g(i), h(i)$. Use the permutation function $\pi_1(g(i))=1$ and $\pi_2(h(i))=1$ so the array has the same distribution as before due to separate exchangeability. Since the array is dissociated, by setting $G_0 =H_0=1$, $W_i$ is independent of all observations that are not in $g(i)$ or $h(i)$, verifying (a). 

For (b), take any $i$ and $k \in \mathcal{N}_i$. Without loss of generality, suppose that $g(i) = g(k)$. Consider the case where $h(i) \ne h(k)$. Use the permutation function $\pi_1(g(i))=1$ and $\pi_2(h(i))=1,\pi_2(h(k))=2$ to get another array that has the same distribution. Since the array is dissociated, by setting $G_0=1,H_0=2$, $(W_i, W_k)$ is independent of all observations that are not in $(\mathcal{N}_i \cup \mathcal{N}_k)$. Since $j,l \notin (\mathcal{N}_i \cup \mathcal{N}_k)$, $(W_i, W_k) \indep (W_j, W_l)$, which yields (b). If $h(k)=h(i)$, set $\pi_2(h(k))=1$ and $G_0=1,H_0=1$. The same argument applies.
\end{proof}

For Proposition 2, I first prove a consistency result.

\begin{lemma} \label{lem:XX}
Under Assumptions \ref{asmp:indep}, \ref{asmp:regularity}(a) and \ref{asmp:regularity}(b), and $E[W_i] =0$ $\forall i$, $ || (1/ n \sum_i (W_i W_i' - E[W_i W_i']) || \xrightarrow{p} 0$. 
\end{lemma}
\begin{proof}
It suffices to show convergence elementwise. Let $X_i$ and $Y_i$ denote scalar components of $W_i$, i.e., $X_i = W_{im}, Y_i = W_{ip}$, where $m,p \in \{1, 2, \cdots, K \}$. By Chebyshev's inequality, and Assumption \ref{asmp:regularity}(a) that $\max_{m,k} E[W_{mk}^4] < K_0$,
\begin{align*}
    P&\left( \frac{1}{n} \sum_i (X_i Y_i - E[X_i Y_i]) > \epsilon \right) \\
    &\leq \frac{1}{\epsilon^2} \frac{1}{n^2} E \left( \sum_i \sum_{j \in \mathcal{N}_i}  (X_i Y_i - E[X_i Y_i])  (X_j Y_j - E[X_j Y_j]) \right) \leq \frac{K_0}{\epsilon^2 n^2} \sum_i \sum_{j \in \mathcal{N}_i} 1
\end{align*}
Hence, it suffices to show $ (\sum_i \sum_{j \in \mathcal{N}_i} 1 ) / n^2 = o(1)$. Observe 
$$\frac{\sum_i \sum_{j \in \mathcal{N}_i} 1}{n^2} \leq \frac{\max_i N_i}{n} \frac{\left( \sum_i 1 \right)}{n}$$
so it suffices to show $ \max_i N_i / n = o(1)$. Since 
$$\lambda_n \leq \sum_i \sum_{j \in \mathcal{N}_i}  \max_m E[W_{mk}^2] \leq n^2 \max_m E[W_{mk}^2],$$
\begin{align*}
    \frac{\left( \max_i N_i\right)^2}{n^2} = \frac{ \left(  \max_i N_i \right)^2 \max_m E[W_{mk}^2]}{n^2 \max_m E[W_{mk}^2]} \leq \max_m E[W_{mk}^2] \frac{\left( \max_i N_i \right)^2}{\lambda_n} = o(1)
\end{align*}
Convergence occurs due to Assumption \ref{asmp:regularity}(b) and $\max_m E[W_{mk}^2] < K_0$.
\end{proof}

\begin{proof} [Proof of Proposition \ref{prop:stochasticX}]
Since $E[u_i^4|X_i] \leq K_0$, $E[u_i^4 X_{ik}^4] = E[E[u_i^4|X_i] X_{ik}^4] \leq K_0 E[X_{ik}^4] \leq K_0^2$ is bounded. By Theorem \ref{thm:main}, $Q^{-1/2}_n \sum_i X_i u_i  \xrightarrow{d} N(0, I_{K})$. 

To complete the normality result, I show that $S_n^{-1} \hat{S}_n \xrightarrow{p} I_K$, which is the same as showing that $|| S_n^{-1} (\hat{S}_n - S_n)|| \xrightarrow{p} 0$. By applying Lemma \ref{lem:XX}, $(1/n) (\hat{S}_n - S_n) = (1/n) \sum_i (X_i X_i' - E[X_i X_i']) = o_P(1)$. Hence, it suffices that $(S_n/n)^{-1}$ has bounded eigenvalues, i.e., $\lambda_{\min} (S_n/ n) \geq K_1 >0$, which is true by Assumption \ref{asmp:stochasticX}(e). Since $\hat{\beta} - \beta = \hat{S}_n^{-1} \sum_i X_i u_i$, by Slutsky's lemma, $Q_n^{-1/2} S_n (\hat{\beta} - \beta) \xrightarrow{d} N(0, I_{K})$. 

Next, proceed to consistent variance estimation. Showing that $|| Q_n^{-1} \hat{Q}_n - I_K|| = o_P(1)$ is equivalent to showing that, $\forall \mu \in \mathbb{R}^K$, $\mu' \left(Q_n^{-1/2} (\hat{Q}_n  - Q_n) Q_n^{-1/2} \right)\mu = o_P(1)$. 
\begin{align*}
    \hat{Q}_n &:= \sum_i \sum_{j \in \mathcal{N}_i} \hat{u}_i \hat{u}_j X_i X_j' = \sum_i \sum_{j \in \mathcal{N}_i} (u_i - X_i' (\hat{\beta} - \beta)) (u_j - X_j' (\hat{\beta} - \beta)) X_i X_j' \\
    &= \sum_i \sum_{j \in \mathcal{N}_i} u_i u_j X_i X_j' - 2 \left( \sum_i \sum_{j \in \mathcal{N}_i} u_i X_j'(\hat{\beta} - \beta) X_i X_j' \right) + \left( \sum_i \sum_{j \in \mathcal{N}_i} X_i'(\hat{\beta} - \beta) X_j'(\hat{\beta} - \beta) X_i X_j' \right)
\end{align*}

By Theorem \ref{thm:main}, $\mu' Q_n^{-1/2} (\sum_i \sum_{j \in \mathcal{N}_i} u_i u_j X_i X_j' - Q_n) Q_n^{-1/2} \mu = o_P(1)$. Hence, it remains to show:
\begin{align*}
    \left|\left| Q_n^{-1/2} \left[ - 2 \left( \sum_i \sum_{j \in \mathcal{N}_i} u_i X_j'(\hat{\beta} - \beta) X_i X_j' \right) + \left( \sum_i \sum_{j \in \mathcal{N}_i} X_i'(\hat{\beta} - \beta) X_j'(\hat{\beta} - \beta) X_i X_j' \right)  \right]  Q_n^{-1/2} \right|\right| = o_P(1)
\end{align*}

Observe that $X_i' (\hat{\beta} - \beta) = \left( X_i' S_n^{-1} Q_n^{1/2} \right) \left( Q_n^{-1/2} S_n (\hat{\beta} - \beta) \right) = \left( X_i' S_n^{-1} Q_n^{1/2} \right)  (Z_{K} + 1_{K} o_P(1))$, where $1_K$ is a $K$-vector of ones and $Z_K \sim N(0,I_K)$. Hence, addressing the second term,
\begin{align*}
    X_i'(\hat{\beta} - \beta) X_j'(\hat{\beta} - \beta) &= \left( X_i' S_n^{-1} Q_n^{1/2} \right) (Z_{K} + 1_{K} o_P(1))  (Z_{K}  + 1_{K} o_P(1))' \left( X_j' S_n^{-1} Q_n^{1/2} \right)' \\
    &= \left( X_i' S_n^{-1} Q_n^{1/2} \right) (I_{K} O_P(1) + o_P(1)) \left( X_j' S_n^{-1} Q_n^{1/2} \right)' \\
    &= X_i' S_n^{-1} Q_n S_n^{-1} X_j O_P(1) 
\end{align*}
This equality implies
\begin{align*}
     Q_n^{-1/2} &\left( \sum_i \sum_{j \in \mathcal{N}_i} X_i'(\hat{\beta} - \beta) X_j'(\hat{\beta} - \beta) X_i X_j' \right) Q_n^{-1/2}  \\
     &= Q_n^{-1/2} \left( \sum_i \sum_{j \in \mathcal{N}_i} \left(  X_i' S_n^{-1} Q_n S_n^{-1} X_j \right) X_i X_j' \right) Q_n^{-1/2}  O_P(1) \\
    &= \frac{1}{n^2} \left( \frac{1}{\lambda_n} Q_n \right)^{-1/2} \left( \sum_i \sum_{j \in \mathcal{N}_i} \left(  X_i' \left(\frac{1}{n} S_n \right)^{-1} \left(\frac{1}{\lambda_n} Q_n \right) \left(\frac{1}{n} S_n \right)^{-1} X_j \right) X_i X_j' \right) \left( \frac{1}{\lambda_n} Q_n \right)^{-1/2} O_P(1) 
\end{align*}

The eigenvalues of $(Q_n/\lambda_n)$ are bounded. To see this, it suffices to show that there exists $K_0 < \infty$ such that $\lambda_{\max}(Q_n) / \lambda_n \leq K_0$. Due to finite moments, $Q_n := Var(\sum_i X_i) \leq K_0 1_{K \times K} \sum_c (N_c^C)^2$. Since $(\sum_c (N_c^C)^2)/\lambda_n \leq K_0$ by Assumption \ref{asmp:stochasticX}, $\lambda_n K_0 \geq \sum_c (N_c^C)^2$, which implies $\lambda_n \geq (\sum_c (N_c^C)^2) / K_0$. Hence, 
\begin{align*}
    \frac{\lambda_{\max} (Q_n)}{\lambda_n} \leq \frac{\sum_c (N_c^C)^2 K_0}{ \sum_c (N_c^C)^2 \frac{1}{K_0}} = K_0^2 
\end{align*}
Recall that $(S_n/n)^{-1}$ has bounded eigenvalues. The proof of Theorem \ref{thm:main} also showed that $(Q_n/\lambda_n)^{-1}$ has bounded eigenvalues. By using Markov and Minkowski inequalities, and the same argument as the proof of Theorem \ref{thm:main} for $\mu \in \mathbb{R}^K, ||\mu||=1$,
\begin{align*}
    P &\left( \frac{1}{n^2} \left| \mu'\left( \frac{1}{\lambda_n} Q_n \right)^{-1/2} \left( \sum_i \sum_{j \in \mathcal{N}_i} \left(  X_i' \left(\frac{1}{n} S_n \right)^{-1} \left(\frac{1}{\lambda_n} Q_n \right) \left(\frac{1}{n} S_n \right)^{-1} X_j \right) X_i X_j' \right) \left( \frac{1}{\lambda_n} Q_n \right)^{-1/2} \mu \right| > \epsilon \right) \\
    &\leq \frac{1}{n^2 \epsilon} E\left[ \left| \mu' \left( \frac{1}{\lambda_n} Q_n \right)^{-1/2} \left( \sum_i \sum_{j \in \mathcal{N}_i} \left(  X_i' \left(\frac{1}{n} S_n \right)^{-1} \left(\frac{1}{\lambda_n} Q_n \right) \left(\frac{1}{n} S_n \right)^{-1} X_j \right) X_i X_j' \right) \left( \frac{1}{\lambda_n} Q_n \right)^{-1/2} \mu \right| \right] \\
    &\leq \frac{1}{n^2 \epsilon} \sum_i N_i \max_{m,k} E[X_{mk}^4] K_0  \leq \frac{\max_i N_i}{n} \frac{n}{n} K_0\rightarrow 0
\end{align*}
where $K_0 \in \mathbb{R}$ is an arbitrary (finite) constant. Convergence occurs due to Assumption \ref{asmp:stochasticX}(b), which implies $\max_i N_i/n \rightarrow 0$, since $\max_i \sum_{j \in \mathcal{N}_i} N_i / n = o(1)$ in the proof of Lemma \ref{lem:XX}.

Going back to the first term,
\begin{align*}
    Q_n^{-1/2} &\sum_i \sum_{j \in \mathcal{N}_i} u_i X_j'(\hat{\beta} - \beta) X_i X_j' Q_n^{-1/2}=  Q_n^{-1/2} \sum_i \sum_{j \in \mathcal{N}_i} u_i \left( X_i' S_n^{-1} Q_n^{1/2} \right) (Z_{K} + 1_{K} o_P(1)) X_i X_j'  Q_n^{-1/2} \\
    &= \frac{1}{n \sqrt{\lambda_n}} \left(\frac{1}{\lambda_n} Q_n \right)^{-1/2} \sum_i \sum_{j \in \mathcal{N}_i} u_i \left( X_i' \left( \frac{1}{n} S_n \right)^{-1} \left( \frac{1}{\lambda_n} Q_n \right)^{1/2} \right)  X_i X_j' \left(\frac{1}{\lambda_n} Q_n \right)^{-1/2} O_P(1)
\end{align*}

By using Markov and Minkowski inequalities,
\begin{align*}
    P &\left( \frac{1}{n \sqrt{\lambda_n}} \left| \mu' \left(\frac{1}{\lambda_n} Q_n \right)^{-1/2} \sum_i \sum_{j \in \mathcal{N}_i} u_i \left( X_i' \left( \frac{1}{n} S_n \right)^{-1} \left( \frac{1}{\lambda_n} Q_n \right)^{1/2} \right)  X_i X_j' \left(\frac{1}{\lambda_n} Q_n \right)^{-1/2} \mu \right| > \epsilon \right) \\
    &\leq \frac{1}{n \sqrt{\lambda_n} \epsilon} E\left[ \left| \mu' \left(\frac{1}{\lambda_n} Q_n \right)^{-1/2} \sum_i \sum_{j \in \mathcal{N}_i} u_i \left( X_i' \left( \frac{1}{n} S_n \right)^{-1} \left( \frac{1}{\lambda_n} Q_n \right)^{1/2} \right)  X_i X_j' \left(\frac{1}{\lambda_n} Q_n \right)^{-1/2} \mu \right| \right] \\
    &\leq \frac{1}{n \sqrt{\lambda_n} \epsilon} \sum_i \sum_{j \in \mathcal{N}_i} \max_{m_1,m_2,k} E\left[ \left| X_{m_1k} u_{m_1} X_{m_2}^2 \right| \right] K_0 \\
    &\leq  \frac{1}{n \sqrt{\lambda_n} \epsilon} \sum_i N_i \max_{m_1,m_2,k} E\left[ | X_{m_1k} u_{m_1} |^{2} \right]^{1/2} E\left[ |X_{m_2}^2 |^{2} \right]^{1/2} K_0 \\
    &\leq \frac{\max_i N_i}{\sqrt{\lambda_n}} \frac{1}{\epsilon} \max_{m_1,m_2,k} E[X_{m_1k}^2 u_{m_1}^2]^{1/2} E[X_{m_2}^4]^{1/2} K_0 = o(1)
\end{align*}

The penultimate inequality occurs due to H\"{o}lder's inequality. Observe that $ \max_i N_i / \sqrt{\lambda_n} = o(1)$ if and only if $ \max_c (N^C_c)^2 / \lambda_n = o(1)$, which is given by Assumption \ref{asmp:stochasticX}(b). Convergence in the last step occurs because $ \max_i N_i / \sqrt{\lambda_n} = o(1)$, and the moments are finite. 

Hence, it has been shown that $Q_n^{-1} \hat{Q}_n \xrightarrow{p} I_K$. Then, $[S_n^{-1} Q_n S_n^{-1}]^{-1} [\hat{S}_n^{-1} \hat{Q}_n \hat{S}_n^{-1}] \xrightarrow{p} I_K$ by the continuous mapping theorem. 
\end{proof}

\bibliographystyle{ecta}

\bibliography{mwclus.bib}

\end{document}